\documentclass[11pt]{article}

\usepackage{geometry}
\usepackage{graphicx}
\usepackage{amsmath, amssymb, amsfonts, amsthm, float}
\usepackage{enumerate, color, framed, float, multirow}
\usepackage{comment, longtable, caption, subcaption, appendix}
\usepackage[sort,longnamesfirst]{natbib}
\usepackage{setspace, parskip}
\usepackage{placeins}
\usepackage{algorithm}
\usepackage[noend]{algpseudocode}

\makeatletter
\def\BState{\State\hskip-\ALG@thistlm}
\makeatother

\allowdisplaybreaks

\newcommand{\ds}{\displaystyle}

\newcommand{\Var}{\text{Var}}

\newcommand\numberthis{\addtocounter{equation}{1}\tag{\theequation}}

\newtheorem{theorem}{Theorem}

\theoremstyle{remark}

\begin{document}
\onehalfspacing
\title{Efficient Bernoulli factory MCMC for intractable posteriors}

\date{\today}
\author{
	Dootika Vats\\
	Department of Mathematics and Statistics\\
	Indian Institute of Technology Kanpur\\
	\texttt{dootika@iitk.ac.in}
	\and
	Fl{\'a}vio B. Gon{\c c}alves\\
	Department of Statistics\\
	Universidade Federal de Minas Gerais\\
	\texttt{fbgoncalves@est.ufmg.br}
	\and
	Krzysztof {\L}atuszy{\'n}ski\\
	Department of Statistics\\
	University of Warwick\\
	\texttt{K.G.Latuszynski@warwick.ac.uk}
	\and
	Gareth O. Roberts\\
	Department of Statistics\\
	University of Warwick\\
	\texttt{Gareth.O.Roberts@warwick.ac.uk}
}
\maketitle
\begin{abstract}
Accept-reject based Markov chain Monte Carlo (MCMC) algorithms have traditionally {utilised acceptance probabilities that can be explicitly written as} a function of the ratio of the target density at the two contested points. This feature is rendered almost useless in Bayesian posteriors with unknown functional forms. We introduce a new family of MCMC acceptance probabilities that has the distinguishing feature of not being a function of the ratio of the  target density at the two points.
 We present two stable Bernoulli factories that generate events within this class of acceptance probabilities. The efficiency of our methods rely on obtaining reasonable local upper or lower bounds on the target density and we present two classes of problems where such bounds are viable: Bayesian inference for diffusions and MCMC on constrained spaces.  The resulting \textit{portkey Barker's} algorithms are exact and  computationally more efficient that the current state-of-the-art. 
\end{abstract}

\section{Introduction} 
\label{sec:introduction}


Markov chain Monte Carlo (MCMC) is a popular tool for drawing samples from complicated distributions.  For Bayesian posteriors, where MCMC is most often used, the target density is usually available only up to a proportionality constant. That is, the target density, $\pi(x)$ is such that,  $\pi(x) \propto \pi'(x)$, where the functional form of $\pi'$ is usually known. The popular MCMC algorithm of \cite{hast:1970,metr:1953} has been immensely useful as unknown normalising constants of Bayesian posteriors play no role in  the sampler. Given a proposal density $q(x,y)$, a move from $x$ to $y$ in a Metropolis-Hastings (MH) algorithm is accepted with probability
\[
\alpha_{\text{MH}}(x,y) = \min \left\{1, \dfrac{\pi(y) q(y,x)}{\pi(x) q(x,y)} \right\} = \min \left\{1, \dfrac{\pi'(y) q(y,x)}{\pi'(x) q(x,y)} \right\}\,.
\]
{On the other hand, there are many other acceptance probability functions that can be used to ensure the stationarity of $\pi$.
For instance,}
 \cite{barker:1965} proposed a similar algorithm with 
\begin{equation}
\label{eq:barker_symm}
	\alpha_{\text{B}}(x,y) = \dfrac{\pi(y) q(y,x)}{\pi(x) q(x,y) + \pi(y)q(y,x)} = \dfrac{\pi'(y) q(y,x)}{\pi'(x) q(x,y) + \pi'(y) q(y,x)}\,.
\end{equation}
Barker's algorithm is used less frequently than MH due to the result of \cite{peskun:1973}, which establishes that MH is always better than Barker's algorithm, in terms of asymptotic variance of ergodic averages. However, as shown in \cite{latu:roberts:2013}, the variance of ergodic averages from Barker's method are no worse than twice that of MH. For this reason, when the MH algorithm is difficult to implement, Barker's algorithm can be particularly important.

\cite{gon:krzy:rob:2017,gon:lat:rob} used Barker's acceptance probability for Bayesian posteriors with unknown functional forms. Such intractable posteriors can be due to intractable priors on constrained parameter spaces or intractable likelihoods for  complex systems which cannot be conveniently modeled tractably. In either case, the Bayesian posterior is not available up to a normalising constant and
$\pi(y)/\pi(x)$ cannot be evaluated making it difficult to implement the MH algorithm. There are various solutions to this proposed in the literature; from inexact algorithms like the double Metropolis-Hastings \citep{lia:liu:car:2007} to exact ones like pseudo-marginal MCMC \citep{and:rob:2009}. See \cite{park:haran} for a comprehensive review.

Nearly all methods modify the Markovian dynamics of an underlying accept-reject based MCMC algorithm. We, instead, focus on the use of Bernoulli factories to avoid explicitly calculating the acceptance probabilities, as in \cite{gon:krzy:rob:2017,gon:lat:rob}.
Their \textit{two-coin} Bernoulli factory generates events of probability $\alpha_{\text{B}}(x,y)$ without explicitly evaluating it. Thus, {an algorithm which evaluates $\alpha_{\text{B}}(x,y)$ explicitly, and one which embeds a Bernoulli factory within each iteration to determine whether to accept a proposed move, are statistically indistinguishable.}
However, as we will demonstrate, the two-coin Bernoulli factory can be computationally burdensome.

This provides strong motivation to consider other acceptance probabilities that yield $\pi$-stationarity and for which efficient Bernoulli factories can be constructed. To date, it has been customary to consider functions  $\alpha (x,y)$
which can be written explicitly in terms of $\pi (y)/\pi (x)$ so that any unknown normalisation constant cancels (see \cite{peskun:1973,tier:1994,bill:diac:2001}). {However,  when using Bernoulli factories we have far more flexibility.}
Keeping this in mind we propose an acceptance probability where for $d(x,y) \geq 0$ such that $d(x,y) = d(y,x)$, a move from $x$ to $y$ is accepted with probability
\[
\alpha(x, y) = \dfrac{\pi(y) q(y,x)}{\pi(x) q(x,y) + \pi(y) q(y,x) + d(x,y)}\,.
\]
We present two choices of $d(x,y)$ for which efficient Bernoulli factories can be constructed to generate events of probability $\alpha(x,y)$.  Naturally, $\alpha(x,y) \leq \alpha_{\text{B}}(x,y)$, so  the variance ordering of \cite{peskun:1973} applies. However, usually $d(x,y)$ will be small, so as to yield little decrease in efficiency.
As we will demonstrate in our examples, any loss in statistical efficiency is made up for in computational efficiency. In addition to a simple illustrative example, Section~\ref{sub:examples} describes two  different directions for  applications of the portkey Barker's method, both stemming from important problems in Bayesian statistics. The first area involves MCMC on constrained spaces where prior normalisation constants are unknown. As an example of this, we give the first exact MCMC method for a well-known Bayesian correlation estimation model. The second  class of applications concerns the Bayesian inference for diffusions, which we illustrate on the Wright-Fisher diffusion model. For all examples, our proposed algorithm leads to a far stable sampling process, with significant computational gains.

\section{Barker's method and the two-coin algorithm} 
\label{sec:bernoulli_factory}

Given a current state of the Markov chain $x$, and a proposal density $q(x,y)$, recall the Barker's acceptance probability in \eqref{eq:barker_symm}. Algorithm~\ref{alg:barker} presents the Barker's update for obtaining
a realization at time $m+1$. Usually, Step~2 is implemented by drawing
$U \sim U[0,1]$ and checking if $U \leq \alpha_{\text{B}}(x_m, y)$. 
However, this is not possible  when $\alpha_{\text{B}}(x_m, y)$ cannot be
evaluated. \cite{gon:krzy:rob:2017,gon:lat:rob} noticed that a Bernoulli factory can be constructed to obtain events of probability $\alpha_{\text{B}}(x_m,y)$ without explicitly evaluating it.
\begin{algorithm}[!h]
\caption{Barker's MCMC for $x_{m+1}$}\label{alg:barker}
\begin{algorithmic}[1]
\State Draw $y \sim q(x_m,dy)$
\State Draw $A \sim \text{Bern}\left(\alpha_{\text{B}}(x_m, y) \right)$
\If {$A = 1$} 
\State $x_{m+1} = y$
\EndIf
\If {$A = 0$} 
\State $x_{m+1} = x_{m}$
\EndIf
\end{algorithmic}
\end{algorithm}


The Bernoulli factory problem is one in which given events that occur with probability $p$, the goal is to simulate an event with probability $h(p$), for some function of interest, $h$ \citep{asmu:glynn:1992,keane:obrie:1994,MR2115037,MR2829311,huber:2017,morina2019bernoulli}. \cite{gon:krzy:rob:2017,gon:lat:rob} proposed the following Bernoulli factory to sample events with probability $\alpha_{\text{B}}(x,y)$. Suppose,
\[
\pi(x) q(x,y) = c_x p_x\,,
\]
where $c_x$ is possibly known and $0 < p_x < 1$. Similarly, $\pi(y) q(y,x) = c_y p_y$. The roles of $c_x$ and $c_y$ are to ensure that $p_x$ and $p_y$ are valid probabilities. We stress that the bound $c_x$ can be a local bound; a global bound over the full support is not required. For asymmetric proposal distributions, $c_x$ and $p_x$ depend on both $x$ on $y$, however we suppress the dependency on $y$ for notational convenience. One way to arrive at $c_x$ and $p_x$ is to find $c_x$ such that 
\begin{equation}
\label{eq:cxpx}
\pi(x) q(x,y) \le  c_x
\hbox{ and then set }p_x = {\pi(x) q(x,y) c_x^{-1}}, 
\end{equation}
with analogous statements for $c_y$ and $p_y$.
%
The two-coin Bernoulli factory of \cite{gon:krzy:rob:2017,gon:lat:rob} presented in Algorithm~\ref{alg:alpha1}, returns events of probability $\alpha_{\text{B}}(x,y)$ with
\[
h(p_x,p_y) := \dfrac{c_yp_y}{c_xp_x + c_yp_y} = \alpha_{\text{B}}(x,y)\,.
\]

%

\begin{algorithm}[!h]
\caption{two-coin algorithm for $\alpha_{\text{B}}(x,y)$}\label{alg:alpha1}
\begin{algorithmic}[1]
\State Draw $C_1 \sim $ \text{Bern}$\left(\dfrac{c_y}{c_y + c_x}\right)$
\If {$C_1 = 1$}
\State Draw $C_2 \sim \text{Bern}(p_y)$
\If {$C_2 = 1$} 
\State output 1
\EndIf
\If {$C_2 = 0$} 
\State {go to Step 1}
\EndIf
\EndIf
\If {$C_1 = 0$}
\State Draw $C_2 \sim \text{Bern}(p_x)$
\If {$C_2 = 1$} 
\State {output 0}
\EndIf
\If {$C_2 = 0$} 
\State{go to Step 1}
\EndIf
\EndIf
\end{algorithmic}
\end{algorithm}

First, both $c_x$ and $c_y$  can be known up to a common normalising constant. Second, it is assumed in Steps 2 and 3 that events of probabilities $p_y$ and $p_x$, respectively, can be simulated. Using laws of conditional probability, it is easy to check that Algorithm~\ref{alg:alpha1} returns 1 with probability $\alpha_{\text{B}}(x,y)$. In addition, as \cite{gon:krzy:rob:2017} describe, the number of loops until the algorithm stops is distributed as a Geom$((c_yp_y + c_xp_x)/(c_y + c_x))$, and the mean execution time is
\[
\dfrac{c_x + c_y}{c_xp_x + c_yp_y} = \dfrac{c_x + c_y}{\pi(x) q(x,y) + \pi(y) q(y,x)}\,.
\]
Clearly, the computational efficiency of the two-coin algorithm relies heavily on the upper bounds $c_x$ and $c_y$. If the bound is loose, then the algorithm yields a large mean execution time.


\section{Portkey Barker's method} 
\label{sec:portkey_barker}

The main source of inefficiency in implementing Barker's method via Bernoulli factories is the inefficiency of the two-coin algorithm. Motivated by this, we introduce a new family of acceptance probabilities and provide an efficient Bernoulli factory for members of this family. For a proposal density $q(x,y)$, consider accepting a proposed value $y$ with probability,
\begin{equation}
\label{eq:gennonrat}
\alpha(x, y) = \dfrac{\pi(y) q(y,x)}{\pi(x)q(x,y) + \pi(y)q(y,x) + d(x,y)}\,,
\end{equation}
where $d(x,y) = d(y,x) \geq 0$. Then $\alpha(x,y)$ yields a $\pi$-reversible Markov chain and the symmetry of $d(x,y)$ is essential to this.
\begin{theorem}
	\label{thm:combined_rever}
For a proposal density $q(x,y)$, a Markov chain with acceptance probability $\alpha(x,y)$ in \eqref{eq:gennonrat} is
 $\pi$-reversible if and only if $d(x,y) = d(y,x)$.
\end{theorem}
\begin{proof}
An acceptance function yields a $\pi$-reversible Markov chain if and only if
\[
\pi(y)q(y,x)\alpha(y,x) = {\pi(x)q(x,y)}  \alpha(x, y)\,.
\]
Let $d(x,y) = d(y,x)$. Consider,
\begin{align*}
{\pi(x)q(x,y)}  \alpha(x, y) & = 
 \dfrac{{\pi(y)q(y,x)} \, \pi(x)q(x,y)}{\pi(x)q(x,y) + \pi(y)q(y,x) + d(x,y)} 
= {\pi(y)q(y,x)}\alpha(y, x) \,. 
\end{align*}
\vspace{-.4cm} ~
\end{proof}
Naturally, $\alpha(x,y) \leq \alpha_{\text{B}}(x,y)$ and by Peskun's ordering, Barker's method is more efficient. However, for a particular choice of $d(x,y)$ we present a Bernoulli factory that provides significant computational gains, enough to supersede the loss of statistical efficiency. For a user-chosen $0 < \beta \leq 1$, consider the following  \textit{portkey Barker's} acceptance probability:
\begin{equation}
	\label{eqn:accept_beta}
	\alpha_{(\beta)}(x, y) := \dfrac{\pi(y) {q(y,x)}}{\pi(y){q(y,x)} + \pi(x){q(x,y)} + \dfrac{(1 - \beta)}{\beta} (c_x + c_y) }\,,
\end{equation}
{with $c_x$ and $c_y$  given by \eqref{eq:cxpx}. }
To ensure $d(x,y)$ is small, $\beta \approx 1$. 

To yield events of probability $\alpha_{(\beta)}(x, y)$, we modify the two-coin algorithm via, what we call a \textit{portkey}{\footnote{We borrow the word \textit{portkey} from the Harry Potter books by J.K. Rowling. As described on \textit{Pottermore.com}, ``The name ‘portkey’ comes from the French ‘porter’ -- to carry -- and the word ‘key’, in the sense of secret or trick''.}}  method. Our \textit{portkey two-coin} algorithm in Algorithm~\ref{alg:portkey}  introduces a first step in the two-coin algorithm that allows immediate rejections with probability $1 - \beta$. For a given proposal, whenever Algorithm~\ref{alg:portkey} loops on Steps 4a and 5a, a Bern$(\beta)$ event is drawn, which if zero, rejects the proposed value immediately. Running Algorithm~\ref{alg:portkey} with $\beta \approx 1$ avoids the large number of loops often witnessed in Algorithm~\ref{alg:alpha1}.


\begin{algorithm}[h]
\caption{Portkey two-coin algorithm}\label{alg:portkey}
\begin{algorithmic}[1]
\State Draw $S\sim$ Bernoulli$(\beta)$
\If {$S = 0$}
\State output 0
\EndIf
\If {$S = 1$}
\State Draw $C_1 \sim $ \text{Bern}$\left(\dfrac{c_y}{c_x + c_y}\right)$
\If {$C_1 = 1$}
\State Draw $C_2 \sim \text{Bern}(p_y)$
\If {$C_2 = 1$} 
	\State {output 1}
\EndIf
\If {$C_2 = 0$} 
	\State{go to Step 1}
\EndIf
\EndIf
\If {$C_1 = 0$}
\State Draw $C_2 \sim \text{Bern}(p_x)$
\If {$C_2 = 1$} 
	\State {output {0}}
\EndIf
\If {$C_2 = 0$} 
	\State {go to Step 1}
\EndIf
\EndIf
\EndIf
\end{algorithmic}
\end{algorithm}

\begin{theorem}
\label{thm:accept_combined}
Algorithm~\ref{alg:portkey} yields output 1 with probability $\alpha_{(\beta)}(x, y)$.
\end{theorem}
\begin{proof}
Let $r$ be the probability of no output in any given loop of the algorithm. Then,
\begin{align*}
r &= \beta \dfrac{c_y(1 - p_y) + c_x(1 - p_x)}{c_x + c_y} \\ 
\Rightarrow \ds \sum_{i=0}^{\infty} r^i & =  \dfrac{c_x + c_y}{(1-\beta)(c_x + c_y) + \beta (c_xp_x + c_yp_y)}\,.
\end{align*}
For the algorithm to output 1, for all $i$, there should be no output in all loops up to $i$, and the $i$th loop should output 1. Thus, the probability that the algorithm outputs 1 is
\[
\beta \dfrac{c_yp_y}{c_x + c_y} \ds \sum_{i=0}^{\infty} r^i = \dfrac{c_yp_y}{ c_xp_x + c_yp_y + \frac{1-\beta}{\beta} (c_x + c_y)}\,.  
\]
\vspace{-.5cm} ~
\end{proof}
%
%
%
%
From Theorem~\ref{thm:combined_rever}, the portkey Barker's algorithm is $\pi$-reversible, and if $\beta > 0$, it is $\pi$-ergodic since $\alpha_{(\beta)}(x,y) > 0$. It may be intuitive to see that Algorithm~\ref{alg:portkey} will quite obviously lead to a smaller mean execution time. More specifically, the number of loops until the algorithm stops is distributed according to a Geom$(s_{\beta})$, where
\[
s_{\beta} = (1 - \beta) + \beta \cdot \dfrac{c_yp_y + c_xp_x}{c_x + c_y}\,.
\]
We note that for $0 < \beta < 1$, $s_{\beta} > 1- \beta$, which is bounded away from zero. Thus, the mean execution time is bounded above. A similar argument cannot be made for the original two-coin algorithm. Specifically, the ratio of the mean execution time of the two-coin algorithm to the portkey two-coin algorithm is
\begin{align*}
\dfrac{1/s_1}{1/s_{\beta}} & = \dfrac{s_{\beta}}{s_1}= (1- \beta) \cdot \left(\dfrac{c_yp_y + c_xp_x}{c_x + c_y}  \right)^{-1} + \beta \,\,\geq \,1\,.
\end{align*}
Thus, if $ (c_yp_y + c_xp_x)/(c_x + c_y) \approx 0$, i.e.,  the original two-coin algorithm is highly inefficient, the ratio diverges to infinity. On the other hand, if $(c_yp_y + c_xp_x)/(c_x + c_y) \approx 1$, i.e., the original two-coin algorithm is efficient, the two algorithms have comparable expected number of loops.

Computational efficiency gained here is at the cost of statistical efficiency. Let $P_{(\beta)}$ and $P_{\text{B}}$ denote the Markov operators for portkey Barker's and Barker's algorithms. In addition, for a function $g$, let $\bar{g}_n$ denote the Monte Carlo estimator of $\int g \pi(dx)$ obtained using a Markov kernel $P$ and denote $\text{var}(g, P):= \lim_{n \to \infty} n\Var_{\pi}(\bar{g}_n)$.

\begin{theorem}
  	\label{thm:barker_mod_ordering}
For $0 < \beta \leq 1$, $  	\alpha_{(\beta)}(x, y) \leq \beta\,  \alpha_{\text{B}}(x, y)\,.$
%
 As a consequence,
 \begin{equation*}
 \label{eq:variance_bark_mod}
 	\text{var}(g, P_{\text{B}}) \leq \beta \,\text{var}(g, P_{(\beta)}) + (\beta - 1) \Var_{\pi}(g)\,.
 \end{equation*}
  \end{theorem}

\begin{proof}
Since $c_x + c_y \geq \pi(x)q(x,y) + \pi(y)q(y,x)$
\begin{align*}
	 \alpha_{(\beta)}(x, y) & = \dfrac{\pi(y)q(y,x) }{\pi(x)q(x,y) + \pi(y)q(y,x) + \beta^{-1}{(1 - \beta)} (c_x + c_y) } \\
	& \leq \beta \cdot\dfrac{\pi(y)q(y,x) }{\beta\pi(x) q(x,y)+ \beta \pi(y) q(y,x) + (1 - \beta) (\pi(x) {q(x,y)}+ \pi(y) {q(y,x)}) } \\
	& = \beta \cdot \alpha_{\text{B}}(x, y)\,.
\end{align*}
Combining \citet[Corollary 1]{latu:roberts:2013} and the ordering of 
\cite{peskun:1973},  or equivalently
\citet[Theorem 2]{zane:2019}, yields,
\[
\text{var}(g, P_{\text{B}}) \leq \beta \text{var}(g, P_{(\beta)}) + ({\beta}- 1) \Var_{\pi}(g)\,. 
\]
\vspace{-.5cm}~
\end{proof}
A variance bound in the opposite direction can be obtained under specific conditions.

\begin{theorem}
	\label{thm:ordering_right_direc}
For $0 < \beta \leq 1$, if there exists $\delta > 0$ such that $p_x > \delta$ and $p_y > \delta$, then
\[
\alpha_{\text{B}}(x,y) \leq \left( 1 + \dfrac{1-\beta}{\delta \beta} \right) \cdot \alpha_{(\beta)}(x, y)\,.
\]
As a consequence,
\[
\text{var} \left(g, P_{(\beta)}\right) \leq \left( 1 + \dfrac{1-\beta}{\delta \beta} \right) \text{var}(g, P_{\text{B}}) + \dfrac{1 -\beta}{\delta \beta} \Var_{\pi}(g)\,.
\]
\end{theorem}
\begin{proof}
Since $p_x \geq \delta$ and $p_y \geq \delta$, $c_x + c_y \leq (\pi(x)q(x,y) + \pi(y)q(y,x))/\delta$. So,
\begin{align*}
	\alpha_{(\beta)}(x,y) & 
	\geq \dfrac{\pi(y) q(y,x)}{\pi(x)q(x,y) + \pi(y)q(y,x) + \dfrac{1-\beta}{\delta \beta} (\pi(x)q(x,y) + \pi(y)q(y,x)) } \\
	&  = \left(1 + \dfrac{1 - \beta}{ \delta \beta} \right)^{-1} \alpha_{\text{B}}(x,y)\,.
\end{align*}
The variance ordering follows from \citet[Corollary
1]{latu:roberts:2013}.
\end{proof}

When $p_x$ or $p_y$ is small, it is desirable to set $\beta$ to be large. Theorem~\ref{thm:ordering_right_direc} suggests that setting $\beta = 1 - \delta$ is desirable. However, $\delta$ will typically not be available.

\section{Flipped portkey two-coin algorithm} 
\label{sec:flipped_two_coin_algorithm}

A significant challenge in implementing the original two-coin algorithm is identifying a suitable $c_x$.  When faced with a target density that is restricted to a constrained space (as in Section~\ref{sub:bayesian_correlation_coefficient}), it may be easier to lower bound $\pi(x) q(x,y)$, or in other words, upper bound $\pi(x)^{-1} q(x,y)^{-1}$. That is, suppose for $\tilde{c}_x > 0$ and $0 < \tilde{p}_x < 1$,
\[
\pi(x)^{-1}q(x,y)^{-1} = \tilde{c}_x \tilde{p}_x\,.
\]
Consider an acceptance probability denoted by $\alpha_{f, (\beta)}(x,y)$ of the form \eqref{eq:gennonrat}, with 
\begin{equation}
\label{eq:d_tilde}
d(x,y) := \tilde{d}(x,y) = \dfrac{(1-\beta)}{\beta} \dfrac{  (\tilde{c}_x + \tilde{c}_y)}{\tilde{c}_x\tilde{p}_x \tilde{c}_y  \tilde{p}_y}\,,	
\end{equation}
Algorithm~\ref{alg:flip_portkey} presents a \textit{flipped portkey} two-coin algorithm for $\alpha_{f, (\beta)}(x,y)$.


\begin{algorithm}[h]
\caption{Flipped portkey two-coin algorithm}\label{alg:flip_portkey}
\begin{algorithmic}[1]
\State Draw $S\sim$ Bernoulli$(\beta)$
\If {$S = 0$}
\State output 0
\EndIf
\If {$S = 1$}
\State Draw $C_1 \sim $ \text{Bern}$\left(\dfrac{\tilde{c}_x}{\tilde{c}_x + \tilde{c}_y}\right)$
\If {$C_1 = 1$}
\State Draw $C_2 \sim \text{Bern}(\tilde{p}_x)$
\If {$C_2 = 1$} 
	\State{output 1}
\EndIf
\If {$C_2 = 0$} 
	\State{go to Step 1}
\EndIf
\EndIf
\If {$C_1 = 0$}
\State Draw $C_2 \sim \text{Bern}(\tilde{p}_y)$
\If {$C_2 = 1$} 
	\State{output {0}}
\EndIf
\If {$C_2 = 0$} 
	\State{go to Step 1}
\EndIf
\EndIf
\EndIf
\end{algorithmic}
\end{algorithm}

\begin{theorem}
\label{thm:accept_flipped}
Algorithm~\ref{alg:flip_portkey} yields output 1 with probability 
\[
\alpha_{f,(\beta)}(x,y) = \dfrac{\pi(y) q(y,x)}{\pi(y) q(y,x) + \pi(x) q(x,y) + \dfrac{(1-\beta)}{\beta} \dfrac{  (\tilde{c}_x + \tilde{c}_y)}{\tilde{c}_x\tilde{p}_x \tilde{c}_y  \tilde{p}_y} }
\]
\end{theorem}
\begin{proof}
Using $\pi(x)^{-1} q(x,y)^{-1} = \tilde{c}_x \tilde{p}_x $ and following the steps of Theorem~\ref{thm:accept_combined},
\begin{align*}
 \Pr(\,\text{output = 1})&= \dfrac{\tilde{c}_x \tilde{p}_x }{\tilde{c}_x \tilde{p}_x + \tilde{c}_y \tilde{p}_y + \beta^{-1}{(1-\beta)} (\tilde{c}_x + \tilde{c}_y)} = \alpha_{f,(\beta)}(x,y)\,.
\end{align*}
 \end{proof}

There are two important consequences of Theorem~\ref{thm:accept_flipped}. First, by a simple application of Theorem~\ref{thm:combined_rever}, Algorithm~\ref{alg:flip_portkey} yields a $\pi$-reversible Markov chain. Second, setting $\beta = 1$, yields the usual Barker's acceptance probability, providing a second Bernoulli factory for Barker's algorithm! 
Variance ordering results for the flipped portkey Barker's acceptance probability can be obtained as before and is given below for completeness.
\begin{theorem}
  	\label{thm:flip_barker_mod_ordering}
For $0 < \beta \leq 1$, $  	\alpha_{f,(\beta)}(x, y) \leq \beta\,  \alpha_{\text{B}}(x, y)\,.$
%
 As a consequence,
 \begin{equation*}
 \label{eq:variance_bark_mod}
 	\text{var}(g, P_{\text{B}}) \leq \beta \,\text{var}(g, P_{f,(\beta)}) + (\beta - 1) \Var_{\pi}(g)\,.
 \end{equation*}
Further, if there exists $\delta > 0$ such that $\tilde{p}_x > \delta$ and $\tilde{p}_y > \delta$, then
\[
\alpha_{\text{B}}(x,y) \leq \left( 1 + \dfrac{1-\beta}{\delta \beta} \right) \cdot \alpha_{f,(\beta)}(x, y)\,.
\]
As a consequence,
\[
\text{var} \left(g, P_{f,(\beta)} \right) \leq \left( 1 + \dfrac{1-\beta}{\delta \beta} \right) \text{var}(g, P_{\text{B}}) + \dfrac{1 -\beta}{\delta \beta} \Var_{\pi}(g)\,.
\] 
  \end{theorem}

In addition to the availability of lower bounds on $\pi(x) q(y,x)$, the flipped portkey algorithm should be chosen if generating $\tilde{p}_x$ coins is easier than generating $p_x$ coins. This advantage is demonstrated clearly in Section~\ref{sub:bayesian_correlation_coefficient}.


\section{Examples} 
\label{sub:examples}
\subsection{Gamma mixture of Weibulls} 
\label{sub:gamma_mixture_of_weibulls}

Consider a target distribution of the form $\pi(\theta) = \int \pi(\theta | \lambda)  \, \nu(d\lambda)$
where $\nu$ is a mixing measure on $\lambda$. We assume that $\theta | \lambda \sim $ Weibull$(\lambda, k)$, where $\lambda$ is the scale parameter and $k$ is a known shape parameter. The proposal distribution is the normal distribution centered at the current step with variance .001. It can be shown that, $ \pi(\theta | \lambda)  \leq {k}/({e \theta}):= c_{\theta}$. 
Then $\pi(\theta) = c_{\theta} \cdot \pi(\theta)/c_{\theta}$ and  events of probability $p_{\theta} = \pi(\theta)/c_{\theta}$ can be simulated. Specifically, draw $\lambda \sim \nu$ and $U \sim U[0,1]$ independently. Then $\Pr\{U \leq\pi(\theta | \lambda)/c_{\theta}\} = p_{\theta}$. We set $\nu = $ Gamma$(10,100)$  and $k = 10$. 

 We run Barker's and portkey Barker's algorithms for $10^5$ steps using the two Bernoulli factories for various values of $\beta$.  Trace plots and autocorrelation plots from one such run are shown in Figures~\ref{fig:trace} and \ref{fig:acf}, respectively. As expected, for smaller values of $\beta$, the acceptance probability is smaller, but the autocorrelations do not increase drastically.
\label{sec:toy_example}
 \begin{figure}[!h]
 	\centering
  \vspace{-.2cm}
 	\includegraphics[width=.60\textwidth]{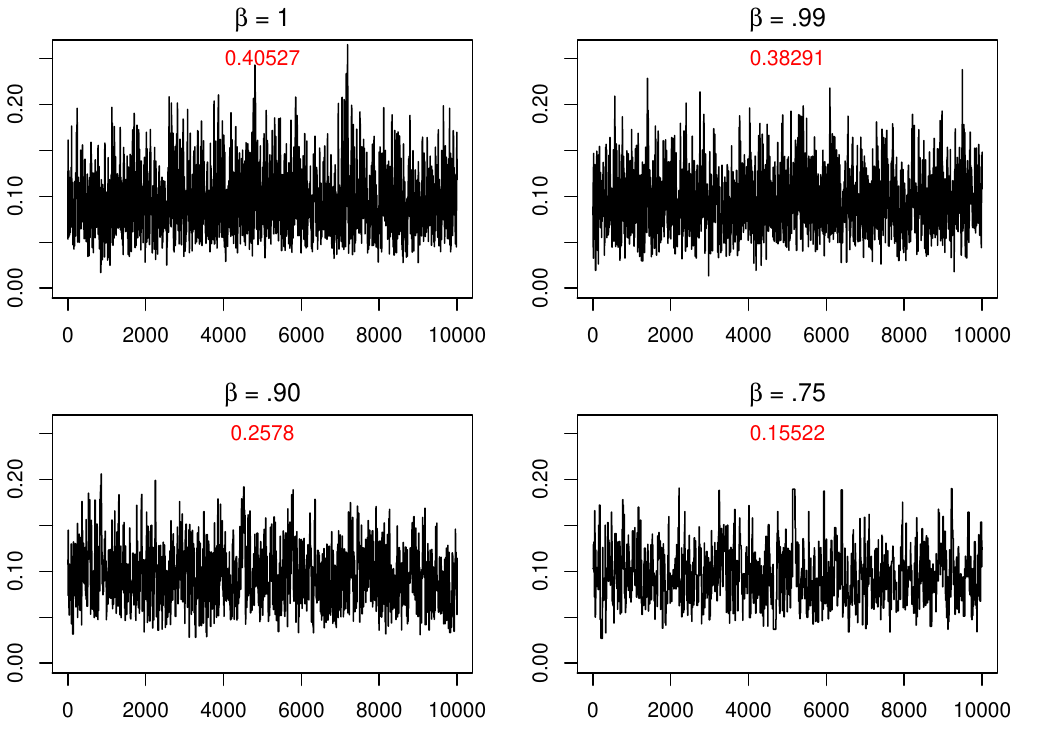} \vspace{-.4cm}
 	\caption{Trace plots of the last 1000 steps of the chains for four values of $\beta$, with acceptance probabilities are in text.}
 	\label{fig:trace}
 \end{figure}
\begin{figure}[!h]
\vspace{-.5cm}
 	\centering
 	\includegraphics[width=0.45\textwidth]{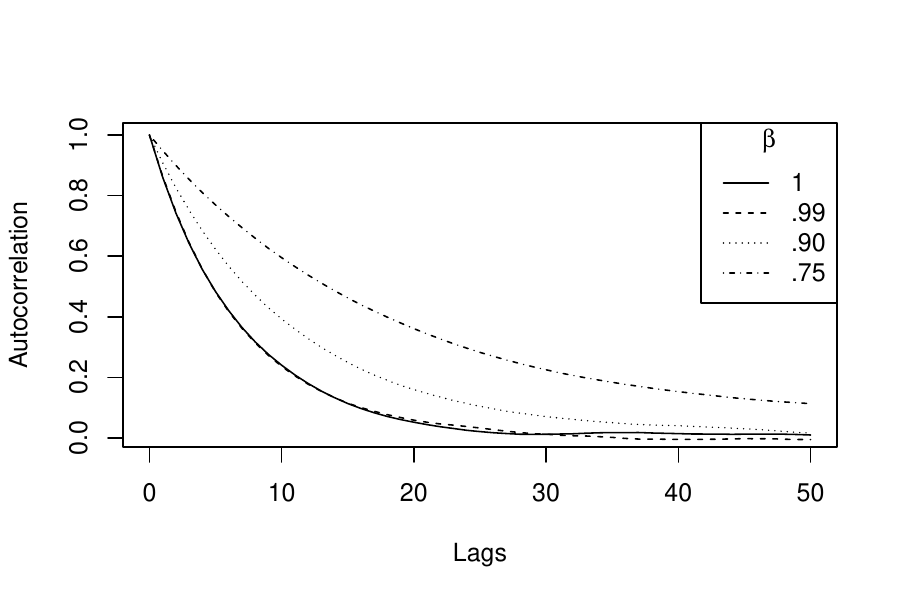} \vspace{-.4cm}
 	\caption{ACF plots for a run-length of 1e5 for four $\beta$ values.}
 	\label{fig:acf}
 \end{figure}

We repeat the simulation 1000 times to  compare their performance. From the results in Table~\ref{tab:mix_results}, unsurprisingly the effective sample size (ESS) \citep{gong:fleg:2015} decreases as $\beta$ decreases. However, the ESS per second is significantly higher for portkey Barker's compared to Barker's. In fact, for $\beta = .90$, the portkey Barker's method is almost three times as efficient as the original Barker's method. This is a clear consequence of the large mean execution time of the two-coin algorithm. The two-coin algorithm required an average of 32 loops, while the portkey  algorithms were significantly smaller. More notably, the two-coin algorithm's loops demonstrates heavy tailed behavior where the average maximum number of loops for MCMC run is 1.3 million.
 \begin{table}[!h]
 	\caption{Averaged results from 1000 replications. Standard errors are in brackets.}
 	\label{tab:mix_results}
 	\centering
 	\begin{tabular}{|l|rrrr|}
 	\hline
 	$\beta $& 1 & .99 & .90 & .75 \\ \hline
ESS & 7484 \tiny{(7.74)} &  6939 \tiny{(12.18)} & 4320 \tiny{(13.86)} & 2501 \tiny{(9.19)} \\
ESS$/s$ & 422.47 \tiny{(3.87)} & 1052.66  \tiny{(2.09)}& 1248.97 \tiny{(4.15)} & 1159.68 \tiny{(4.39)} \\
Mean loops & 32.00 \tiny{(3.7)} &  7.63 \tiny{(0.0)} &  3.97 \tiny{(0.0)}& 2.55 \tiny{(0.0)}\\
Max loops & 1315683 \tiny{(366763.76)}  &   604 \tiny{(3.44)}  &    78   \tiny{(0.37)} &   32  \tiny{(0.12)}\\ \hline
 	\end{tabular}
 \end{table}

\subsection{MCMC on constrained spaces} 
\label{sub:bayesian_correlation_coefficient}

Consider the following general setup of \cite{liech:mull:2009}. Let $f(y |  \theta)$ be a likelihood and for a set $\mathcal{A}$, let $\pi(\theta | \eta)$ be the prior on $\theta$ constrained in $\mathcal{A}$, so that  for a function $g(\cdot | \eta)$,
\[
\pi(\theta | \eta) = \dfrac{g(\theta | \eta) I(\theta \in \mathcal{A})}{\int_{\mathcal{A}} g(\theta | \eta) d \theta}\,,
\]
where $\int_{\mathcal{A}} g(\theta |\eta) d \theta$ is not tractable. If $\eta$ is a fixed hyperparameter, this intractability is not an issue. However, for a full Bayesian model, it is desirable to assign a hyperprior to $\eta$, and this leads to intractability in the posterior distribution for $(\theta, \eta)$. \cite{chen:2010} identify a shadow prior technique that modifies the Bayesian posterior to yield approximate inference. We describe this technique as it relates to a Bayesian correlation estimation model and  demonstrate how the flipped portkey algorithm can be used very naturally in this setting.  

As a general recipe, if one can find a set $\mathcal{B} \supset \mathcal{A}$ such that $\int_{\mathcal{B}} g(\theta | \eta) d\theta$ is tractable, then $\pi(\theta | \eta)$ can be lower bounded easily. The efficiency of the flipped portkey algorithm depends on the relative sizes of sets $\mathcal{B}$ and $\mathcal{A}$. If $g(\theta|\eta)$ is a well-defined density on a larger space, then  $\int_{\mathcal{A}} g(\theta | \eta) d\theta \leq 1$ and depending on the constraint set $\mathcal{A}$, may lead to an efficient flipped portkey algorithm. Such a construction can be done for the  correlation estimation models of \cite{hart:groen:2020,liech:mull:2009,phil:li:2014,wang:pillai:2013}, for the Bayesian graphical lasso \citep{wang2012bayesian}, and for Bayesian semiparametric regression \citep{pap:mar:2020}.

Consider the Bayesian common correlation estimation model of \cite{liec:mull:2004}. Suppose $y_1, \dots, y_n | R \overset{iid}{\sim} N(0, R)$ where $R$ is a $p \times p$ correlation matrix. \cite{liec:mull:2004} assume the unique elements in $R$, $r_{ij}$,  are normally distributed, restricted to $R$ being positive-definite. That is, let $S_p^+$ be the set of $p \times p$ positive-definite matrices, then 
\begin{align*}
f(R \mid \mu, \sigma^2) &= L(\mu, \sigma^2) \prod_{i < j} \dfrac{1}{\sqrt{2 \pi\sigma^2}}\exp \left\{-\dfrac{(r_{ij} - \mu)^2}{2\sigma^2} \right\} \mathbb{I}\{R \in S_p^+ \}\,, \text{ where } \\
L^{-1}(\mu, \sigma^2) &= \int_{R \in S_p^+} \prod_{i < j} \dfrac{1}{\sqrt{2 \pi\sigma^2}}\exp \left\{-\dfrac{(r_{ij} - \mu)^2}{2\sigma^2} \right\} dr_{ij} \numberthis \label{eq:corr_const}
\end{align*}
%
%
%
%
is typically not available in closed form due to the complex nature of $S^+_p$. For known $\tau^2, a_0, b_0$,  hyperpriors for $\mu \sim N(0, \tau^2)$ and $\sigma^2 \sim IG(a_0, b_0)$ are chosen. Interest is in the posterior distribution for $(R, \mu, \sigma^2)$ and \cite{liec:mull:2004} implement a component-wise Metropolis-within-Gibbs sampler. Let $l = p(p-1)/2$. The full conditional densities are
\begin{align*}
f(r_{ij} \mid r_{-ij}, \mu, \sigma^2) &\propto |R|^{-n/2} \exp \left\{-\dfrac{\text{tr}(R^{-1} Y^T Y)}{2}  \right\} \exp \left\{-\dfrac{(r_{ij} - \mu)^2}{2\sigma^2}  \right\} \mathbb{I}\{l_{ij} \leq r_{ij} \leq u_{ij} \}\,,\\ 
f(\mu \mid R, \sigma^2) & \propto L(\mu, \sigma^2) \prod_{i < j} \exp \left\{- \dfrac{(r_{ij} - \mu)^2}{2\sigma^2}  \right\} \exp \left\{- \dfrac{\mu^2}{2\tau^2}  \right\} := L(\mu, \sigma^2) g(\mu , R, \sigma^2) \,, \\ 
f(\sigma^2 \mid R, \mu) & \propto L(\mu, \sigma^2) \prod_{i < j} \exp \left\{- \dfrac{(r_{ij} - \mu)^2}{2\sigma^2}  \right\}  \left(\dfrac{1}{\sigma^2} \right)^{a_0 + l/2 + 1} \exp \left\{ -\dfrac{b_0}{\sigma^2} \right\}\,,
\end{align*}
where the indicator variable  $\mathbb{I}\{l_{ij} \leq r_{ij} \leq u_{ij} \}$ ensures positive-definiteness of $R$. The interval $(l_{ij}, u_{ij})$ can be obtained deterministically using the methods of \cite{barn:mccu:meng:2000}. We use Metropolis-Hastings update with a Gaussian proposal for each $r_{ij}$. 

Updating $\mu$ and $\sigma^2$ requires the knowledge of $L(\mu, \sigma^2)$ which is unavailable. \cite{liec:mull:2004} implement a shadow prior technique that interjects the Bayesian hierarchy such that $r_{ij} \sim N(\delta_{ij}, v^2)$ and $\delta_{ij} \sim N(\mu, \sigma^2)$, with unchanged hyper-priors on $\mu$ and $\sigma^2$. The resulting marginal posterior of $(R, \mu, \sigma^2)$ is different from the original model. The full conditionals of $\mu$ and $\sigma^2$ are available in closed-form, but the full conditionals for $\delta$ are intractable. However, \cite{liech:mull:2009} argue that for updating $\delta_{ij}$, if $v^2 \approx 0$, the intractable constants can be ignored. Thus, the methodology, although convenient, is not asymptotically exact since the resulting Markov chain targets an approximation of a modified desired distribution.


The flipped portkey two-coin algorithm can easily be implemented for both $\mu$ and $\sigma^2$ and we only present details for $\mu$. We use Gaussian random walk proposals for both components so that we can ignore the proposal density in the Bernoulli factories. 
Although $L^{-1}(\mu, \sigma^2) \leq 1$, we can obtain a tighter upper bound by noting that  each $r_{ij} \in [-1,1]$ and thus
\[
  g(\mu , R, \sigma^2)^{-1} L^{-1}(\mu, \sigma^2) \leq g(\mu , R, \sigma^2)^{-1}\left[ \Phi\left(\sigma^{-1}({1 - \mu}) \right) - \Phi\left(\sigma^{-1}({-1 - \mu})\right) \right]^l := \tilde{c}_{\mu}
\]
Set $f(\mu | R, \sigma^2)^{-1}   = \tilde{c}_{\mu} \tilde{c}_{\mu}^{-1}f(\mu \mid R, \sigma^2)^{-1}   := \tilde{c}_{\mu}   \tilde{p}_{\mu}$.  Generating coins of probability $\tilde{p}_{\mu}$ is straightforward:  draw $z_{ij} \sim \text{TN}(-1,1, \mu, \sigma^2)$ for $i < j$, and construct matrix $Z$ with $z_{ij}$ as the lower-triangular entries. If $Z$ is positive-definite, return 1, else return 0. Notice here that the use of the flipped portkey algorithm over the portkey algorithm makes it much easier to simulate the $p$ coin since $L^{-1}(\mu, \sigma^2)$ takes the form of an integral.

Consider the daily closing prices of major European stocks: Germany DAX, Switzerland SMI, France CAC, and United Kingdom FTSE on each day, not including weekends and holidays from 1991-1998. The data are available in the \texttt{datasets} \texttt{R} package  and has 1860 observations. The goal is to estimate the correlation matrix of the four stock prices using the model specified above.

We set $\beta = .90$ for both the updates for $\mu$ and $\sigma^2$ and first compare the estimated posterior density of $\mu$ and $\sigma^2$ from a run of length $10^5$ for both the shadow prior method and the two Bernoulli factory MCMC algorithms. Figure~\ref{fig:shadow} presents the results, where it is evident that the shadow prior method yields a biased posterior distribution.
\begin{figure}[htbp]
	  \vspace{-.5cm}
  \centering
	\includegraphics[width = 4in]{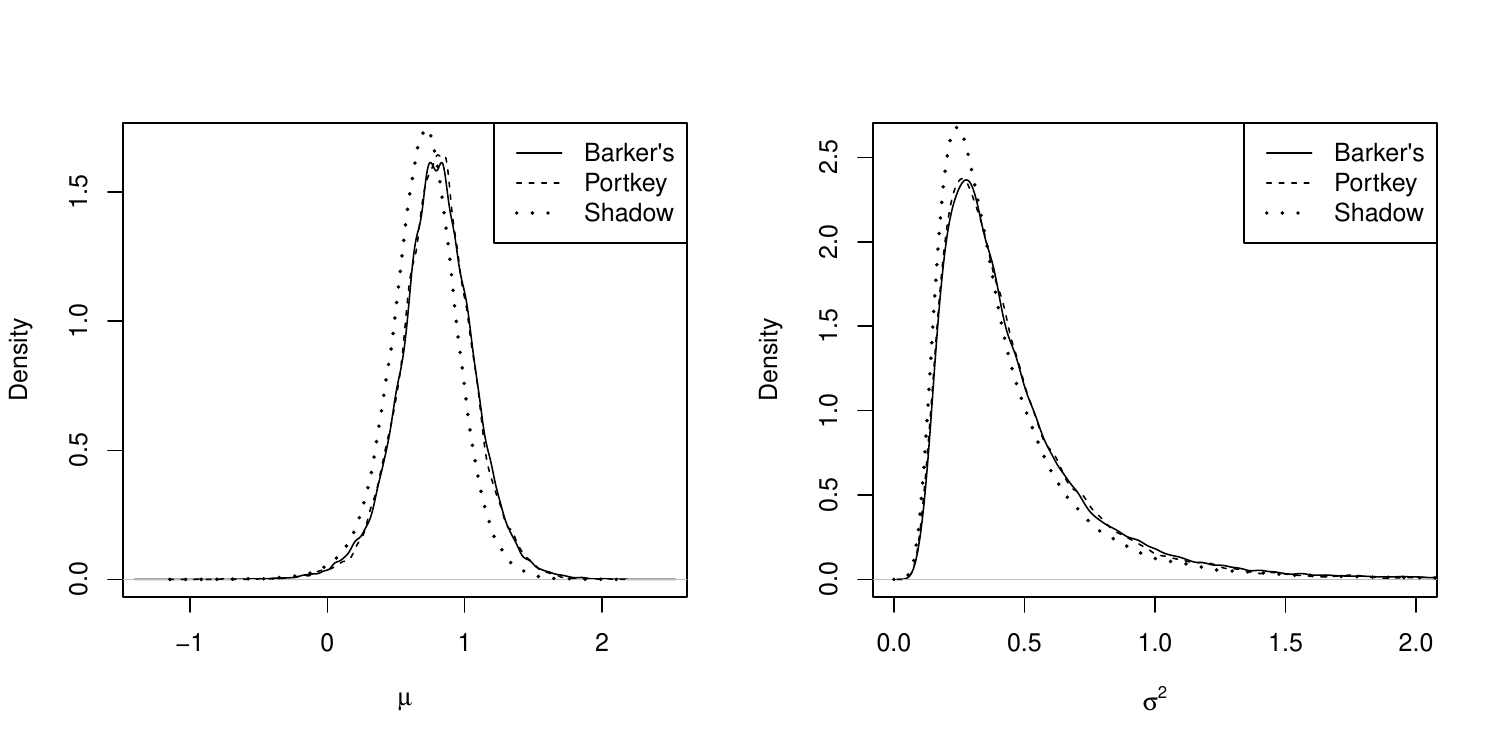}\vspace{-.5cm}
	\caption{Estimated posterior density of $\mu$ and $\sigma^2$ from Markov chain lengths of $10^5$ with $v^2 = .001$ as recommended by \cite{chen:2010}.}
	\label{fig:shadow}
\end{figure}

Figure~\ref{fig:corr_acf} presents the autocorrelation  and trace plots for standard Barker's ($\beta = 1$) and the flipped portkey Barker's ($\beta = .90$). Portkey Barker's algorithm leads to  only slightly slower mixing of the chain. However, the (log of) the number of Bernoulli factory loops required for the portkey two-coin algorithm is far less. Particularly for updating $\mu$, the original two-coin algorithm requires an exponentially large number of Bernoulli factory loops whenever unlikely values are proposed. This makes it laborious to tune Barker's algorithm and encourages short jumps. 
\begin{figure}[htbp]
	\centering
	\vspace{-.6cm} \hspace{-.6cm}
	\includegraphics[width = 1.99in]{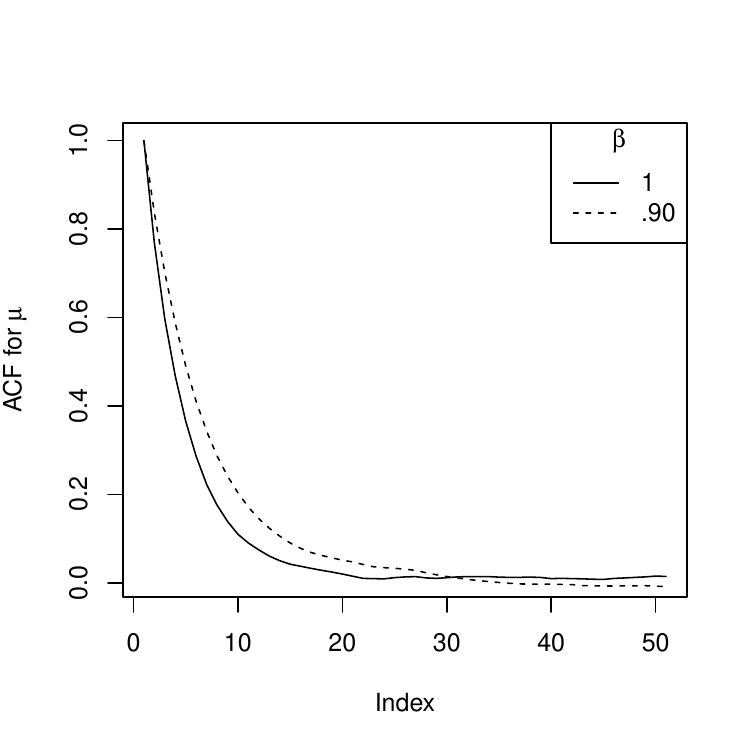} \hspace{-.5cm}
	\includegraphics[width = 1.99in]{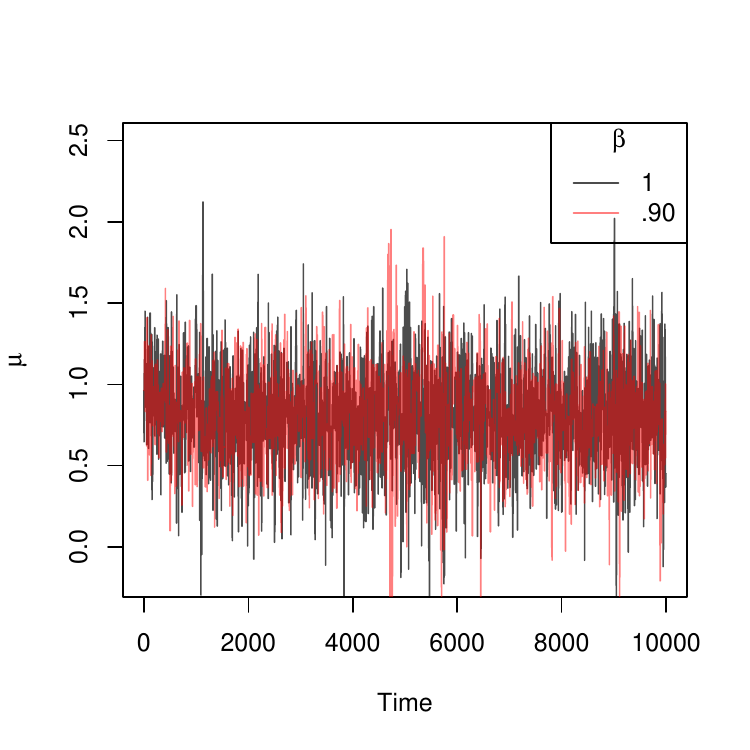}\hspace{-.3cm}
	\includegraphics[width = 1.99in]{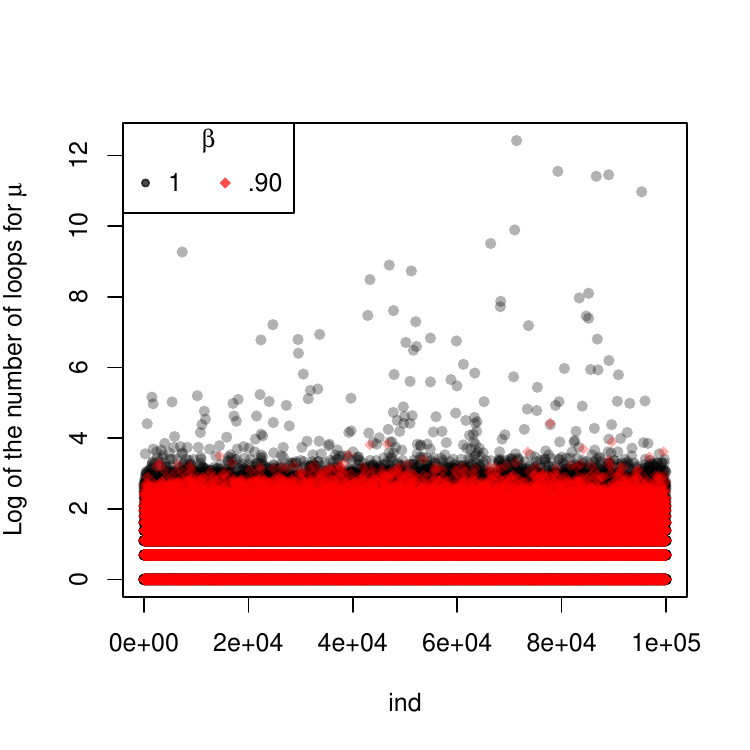}\hspace{-.3cm} \\  \vspace{-.5cm}
	\hspace{-.6cm}
	\includegraphics[width = 1.99in]{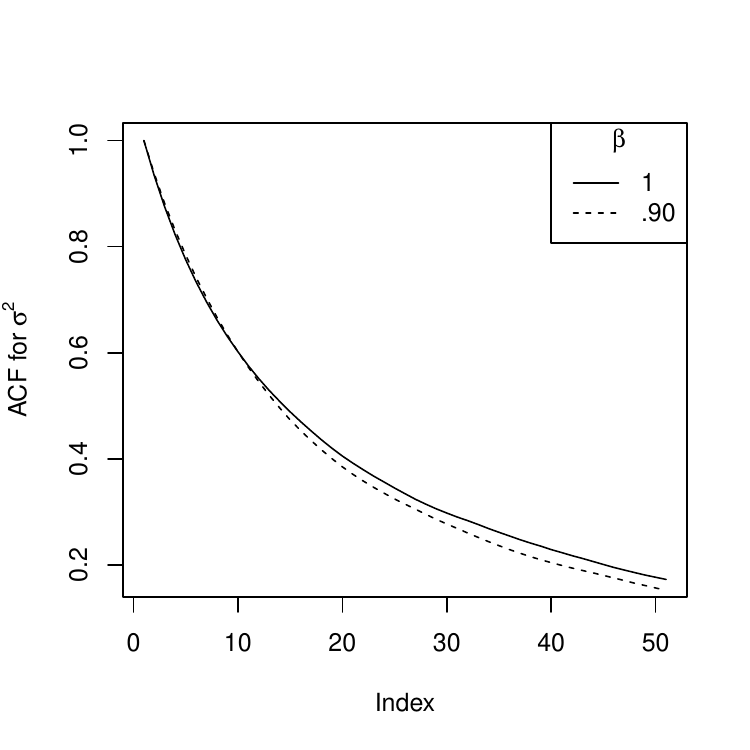} \hspace{-.5cm}
	\includegraphics[width = 1.99in]{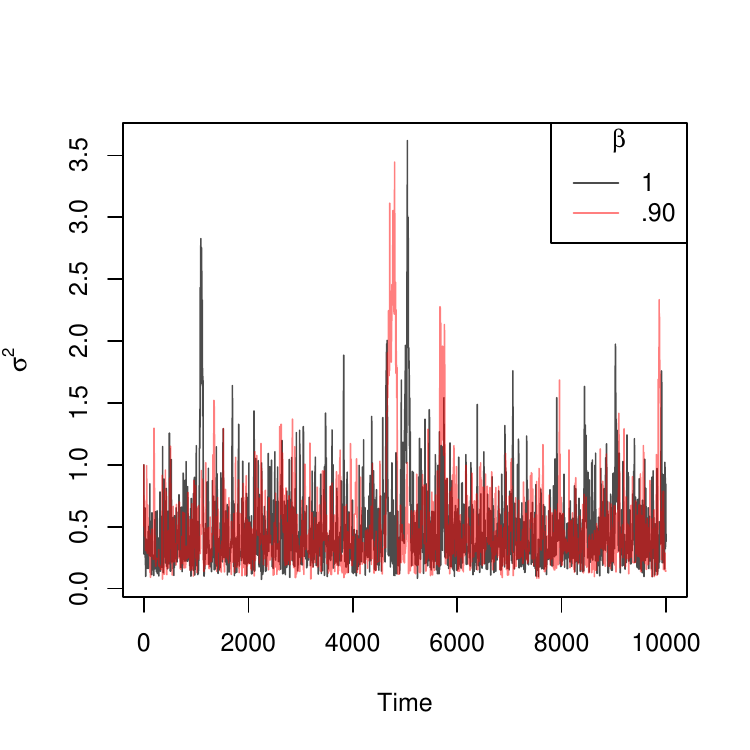}  \hspace{-.3cm}	 
	\includegraphics[width = 1.99in]{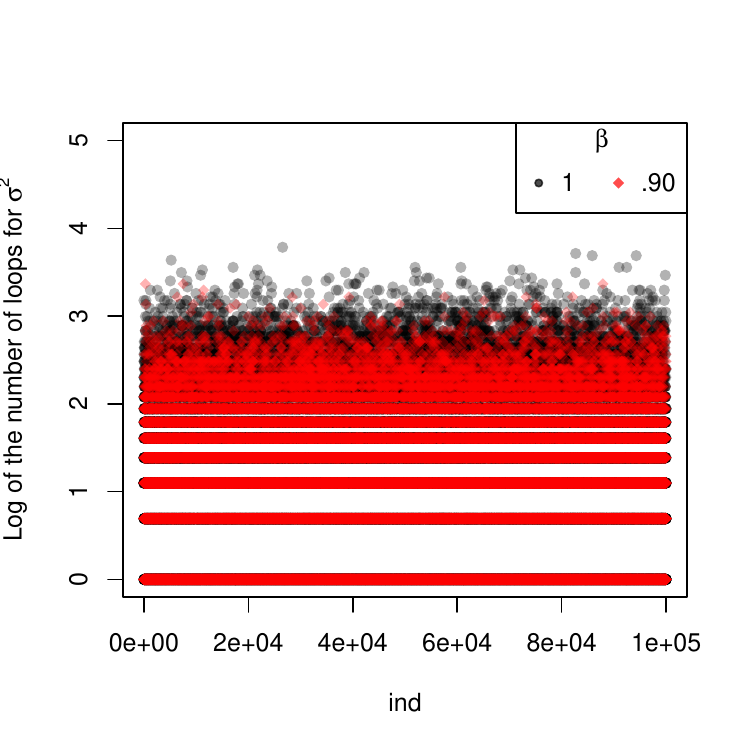}	\hspace{-.3cm}   \vspace{-.5cm} 
	\caption{ACF, trace, and (log of) Bernoulli factory loops for Barker's and portkey Barker's for $\mu$ (right) and $\sigma^2$ (left). }
	\label{fig:corr_acf}
\end{figure}

We repeat the above experiment 10 times for MCMC runs of length $10^4$. Ideally we would have liked to increase the Monte Carlo replication size, but the original two-coin algorithm often got stuck in a loop for days.  The results are presented in Table~\ref{tab:corre_results} with estimated ESS \citep{vats:fleg:jones:2019} calculated using the \texttt{R} package \texttt{mcmcse} \citep{fleg:hugh:vats:2017}. Although a regular Barker's implementation yields a slightly higher ESS, the ESS per second for the flipped portkey Barker's algorithm is about 1.5 times higher on average. As further demonstrated in the table, this is clearly a consequence of the large number of loops required for the two-coin algorithm for updating $\mu$.

 \begin{table}[!h]
 	\caption{Averaged results from 10 replications. Standard errors are  in brackets.}
 	\label{tab:corre_results}
 	\centering
 	\begin{tabular}{|l|rr|}
 	\hline
 	$\beta $& 1 & .90  \\ \hline
ESS & 542 \tiny{(13.50)} &  496 \tiny{(9.00)} \\
ESS$/s$ & 9.63 \tiny{(1.992)} & 14.83  \tiny{(0.279)} \\
Mean loops $\mu$ & 218.43 \tiny{(148.89)} &  2.99 \tiny{(0.010)}\\
Mean loops $\sigma^2$ & 3.21 \tiny{(0.02)} &  2.49 \tiny{(0.010)}\\
Max loops $\mu$ & 2084195 \tiny{(1491777)}  &   34 \tiny{(2.94)} \\ 
Max loops $\sigma^2$ & 38 \tiny{(1.13)}  &   27 \tiny{(1.51)} \\  \hline
 	\end{tabular}
 \end{table}

We stress that for a fair comparison, we use the same proposal distributions for both Markov chains. Since the original two-coin algorithm gets stuck in seemingly infinite loops for  proposals encouraging larger jumps, we are unable to tune the Barker's algorithm to the recommended 15.8\%  acceptance rate of \cite{agr:vats:lat:rob:2021}. To highlight this point, one run of approximately optimally tuned Barker's two-coin algorithm did not produce even $10^3$ samples in 24 hours, whereas the same proposal for the portkey two-coin algorithm yielded $10^4$ draws in $\approx$ 40 seconds with an ESS of 514. 

As described previously, such a Bernoulli factory MCMC algorithm can also be constructed for other models. For instance, consider the Bayesian graphical lasso problem of \cite{wang2012bayesian}. Let $\Omega$ be a $p \times p$ covariance matrix with elements $w_{ij}$; the Bayesian graphical lasso models the observed data as $y_i |\Omega \overset{iid}{\sim} N(0, \Omega^{-1})$ for $i = 1, \dots, n$ with the prior
\begin{align*}
  \Omega \mid \lambda & \sim C_{\lambda}^{-1}  \left\{\prod_{i < j}  \text{DE}(w_{ij} | \lambda_{ij})  \prod_{i = 1}^{p} \text{Exp}(w_{ii} | \lambda_{ii}/2) \right\} \mathbb{I}\{\Omega \in S^+_p\}\,,
\end{align*}
where DE denotes the double exponential distribution. 
For a fully Bayesian paradigm and element-dependent penalization, \cite{wang2012bayesian} discuss the challenges in assigning separate hyper-priors on each of the $\lambda_{ij}$,
in that the normalizing constant $C_{\lambda}$ remains unknown. Since $C_{\lambda} \leq 1$, we can obtain a lower bound on the joint posterior distribution of $(\Omega, \lambda)$ to implement the flipped portkey two-coin algorithms. This bound is likely to not scale well with $p$, so that the efficiency of the Bernoulli factory algorithms will be affected in high-dimensions, however, the data size $n$, will typically have no affect on the computational efficiency.  

In general, the flipped two-coin algorithm with $\beta = 1$ may work reasonably well if the proposal distribution makes small jumps. This, of course, increases the autocorrelation of the resulting Markov chain. If jump sizes are increased, the original two-coin algorithm will require a large number of loops in the Bernoulli factory. The portkey trick stabilizes this behavior, with little loss in statistical efficiency, and often an overall gain in the robustness of the algorithm. 

\subsection{Bayesian inference for the Wright-Fisher diffusion} 
\label{sec:wright_fisher_diffusion}

A collection of methodologies to perform exact inference for discretely-observed diffusions has been proposed in the last 15 years \citep[see][]{bpr06a,AoS,sermai}. These are all based on algorithms for exact simulation of diffusions \citep[][]{bpr07}. Although these methodologies can be applied to a wide class of univariate diffusion models, some important ones are left out. For example, the Cox-Ingersoll-Ross model used to describe the evolution of interest rates, and the Wright-Fisher diffusion, which is widely used in genetics to model the evolution of the frequency of a genetic variant or allele, in a large randomly mating population.

 The methodologies for exact inference for jump-diffusion models  based on the exact simulation of jump-diffusions are considerably more restrict \citep[see][]{gon:lat:rob,flavio2}. Let $\alpha$ be the drift function of the Lamperti transform of the original diffusion. Whilst for diffusion models $\alpha^2+\alpha'$ is required to be bounded below, for jump-diffusions it is also required to have uniformly bounded drift and jump rate of the Lamperti transformed process.

As demonstrated in \cite{gon:krzy:rob:2017,gon:lat:rob}, the two-coin Barker's MCMC algorithm performs exact inference for diffusions and jump-diffusions without requiring the above conditions. The Portkey two-coin Barker's MCMC proposed here will typically provide considerable gain when compared to the original two-coin Barker's, allowing for a wider applicability of the methodology in terms of model complexity and data size.

Consider the neutral Wright-Fisher family of diffusions with mutations
\begin{equation}
\label{eq:wright_fisher}
	dY_s = .5 (\theta_1 (1 - Y_s) - \theta_2 Y_s) + \sqrt{Y_s(1 - Y_s)} dW_s \quad \quad Y_0 = y_0,  \quad \theta_1, \theta_2 > 0\,,
\end{equation}
where $W_s$ is a standard Brownian motion. The process $Y$ lives in
$[0,1]$ and we assume it is observed
at a finite set of time points, $0 = t_0, t_1, \dots, t_n = T$ so that
the observed data are $Y_{\text{obs}} = (y_0,y_1, \dots,
y_{n})^T$. The parameters of interest are $(\theta_1, \theta_2)$ which
describe the drift coefficient of the process and Bayesian estimation yields desirable theoretical properties \citep{sant2020convergence}. Yet, their  transition densities $p(\gamma
|  y_{i-1}, y_{i}):= \Pr_{\gamma}(Y_{t_i} \in dy_i | Y_{t_i - 1} =
y_{i-1})/dy_i$, that determine the likelihoods, are not
available in closed-form and are characterised by  poor
numerical properties near the $0$ and $1$ barriers 
\citep[see][]{jenk:span:2017} that undermine reliable Bayesian inference based on approximations. 



Inference will be done for the following reparameterization $\gamma_1 = \theta_1 + \theta_2$ and $\gamma_2 = \theta_1/(\theta_1 + \theta_2)$,
with uniform priors on $\gamma = (\gamma_1, \gamma_2)$. The first
step is to apply the Lamperti transform to $Y$ in
\eqref{eq:wright_fisher}, which leads to a unit diffusion coefficient
process $X$. Now we specify a Gibbs sampler alternating between updating
 $\gamma$ and updating the missing paths, $X_{\text{mis}}$, of the transformed diffusion $X$ given the transformed observations~$X_{\text{obs}}$.
For the $X_{\text{mis}}$ update, we use standard Brownian bridge
proposals restricted to $[0,1]$ and implement both Barker's and
portkey Barker's MCMC. The Brownian bridges are
 sampled using a layered Brownian bridge construction, which
allows lower and upper bounds on the likelihood that are crucial to
devise an efficient Bernoulli factory. 

For updating $\gamma$, component-wise updates are done with uniform random walk proposal
distributions, $U(\gamma_1 \pm 0.3)$ and $U(\gamma_2 \pm
0.01)$ for $\gamma_1$ and $\gamma_2$, respectively.  
It can be shown that
$
\pi(\gamma |X_{\text{mis}}, X_{\text{obs}}) = c_{\gamma} \, p_{\gamma}\,,
$
where the forms of $c_{\gamma}$ and $p_{\gamma}$ can be found in
\cite{gon:krzy:rob:2017}. A realization of probability $p_{\gamma}$ is
obtained using the layer refinement strategy described in
\cite{gon:lat:rob}. Here each interval $(t_{i-1}, t_{i})$ is broken
into refinements, and a layered Brownian bridge is constructed within
each refinement. Without the refinement, the bound $c_{\gamma}$ is too
loose and thus $p_{\gamma}$ is often too small for the Bernoulli
factory to be efficient. However, finer refinements require simulating
more layered Brownian bridges leading to much larger computation times
to even draw from the proposal distribution. It is precisely this
characteristic of the sampling process that dictates the superior
performance of the portkey two-coin algorithm.

We simulate a Wright-Fisher diffusion with $\gamma_1 = 8$ and $\gamma_2 = 0.5$ and observe the process at times $\{0, 1, \dots, 50\}$. 
We use the component-wise sampler described above to sample from the posterior distribution of $(\gamma_1, \gamma_2)$. We run the two samplers for 10000 steps where the Barker's algorithm takes about 34.5 hours and the portkey Barker's algorithm takes about 3.5 hours.
%
First, the update for $X_{\text{mis}}$ performs similarly for both Barker's and portkey Barker's, where we set $\beta = 0.98$. The average number of loops over $1e4$ samples for the two-coin and the portkey two-coin algorithms are 1.698 and 1.645, respectively, indicating that the original two-coin algorithm already works well enough for this component. We note specifically that close to nothing is lost by using portkey Barker's instead of Barker's acceptance probability.


For updating the parameters, the $\beta$s are set to be .99995 and .9995 for $\gamma_1$ and $\gamma_2$, respectively.
 In Figure~\ref{fig:WF_loops} are the number of loops of the two Bernoulli factories  for each iteration of the two Markov chains for both $\gamma_1$ and $\gamma_2$. The $\gamma_1$ update of the Barker's two-coin Markov chain requires an average of 489 loops with a maximum of 183077 loops. In comparison, the portkey two-coin algorithm runs an average of 43 loops with a maximum of 12622. The long tails of the number of loops for the original two-coin algorithm mean that the average run-time is quite slow, as evidenced from the computation times.
\begin{figure}[htbp]
	  \vspace{-.8cm}
  \centering
	\includegraphics[width = 2.2in]{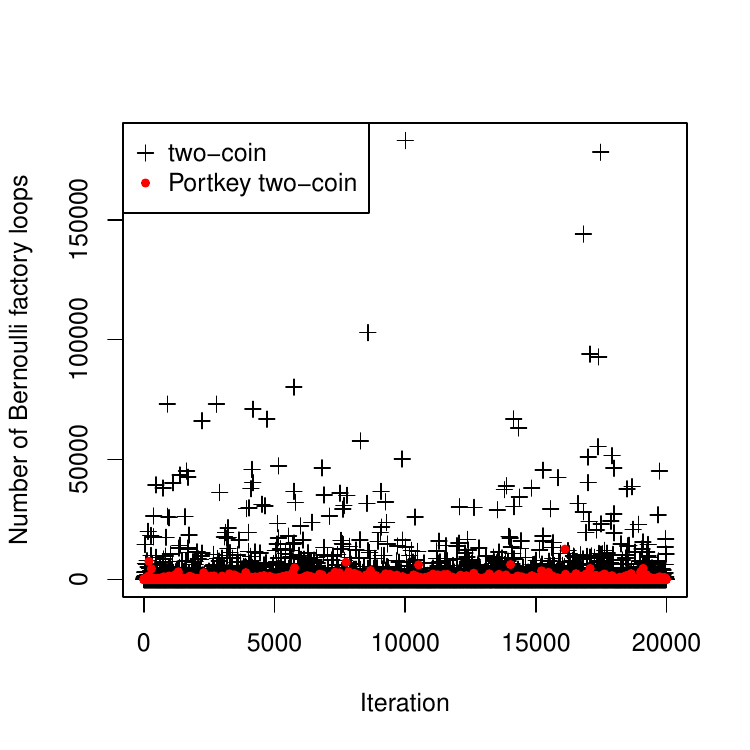}
	\includegraphics[width = 2.2in]{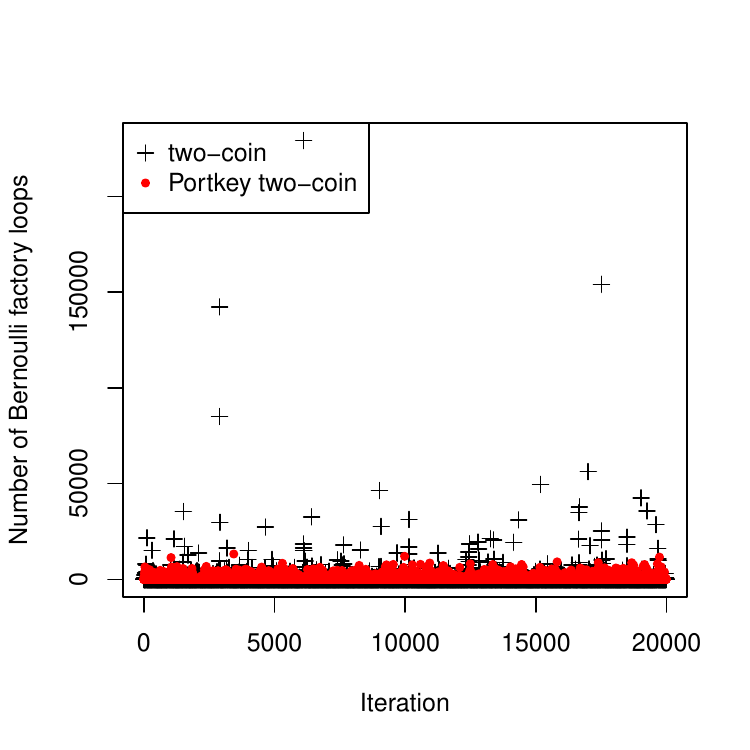} \vspace{-.5cm}
	\caption{Comparing the number of loops of the Bernoulli factory for two-coin and portkey two-coin algorithm. Left is for the $\gamma_1$ update and right is for the $\gamma_2$ update.}
	\label{fig:WF_loops}
\end{figure}
However, since both $\beta$s are close to 1, this gain in computational efficiency mainly impacts the Bernoulli factory, and has only a small effect on the Markov chain as witnessed by the autocorrelation plots in Figure~\ref{fig:wf_acf}.
\begin{figure}[htbp]
  \vspace{-.5cm}
	\centering
	\includegraphics[width = 2.1in]{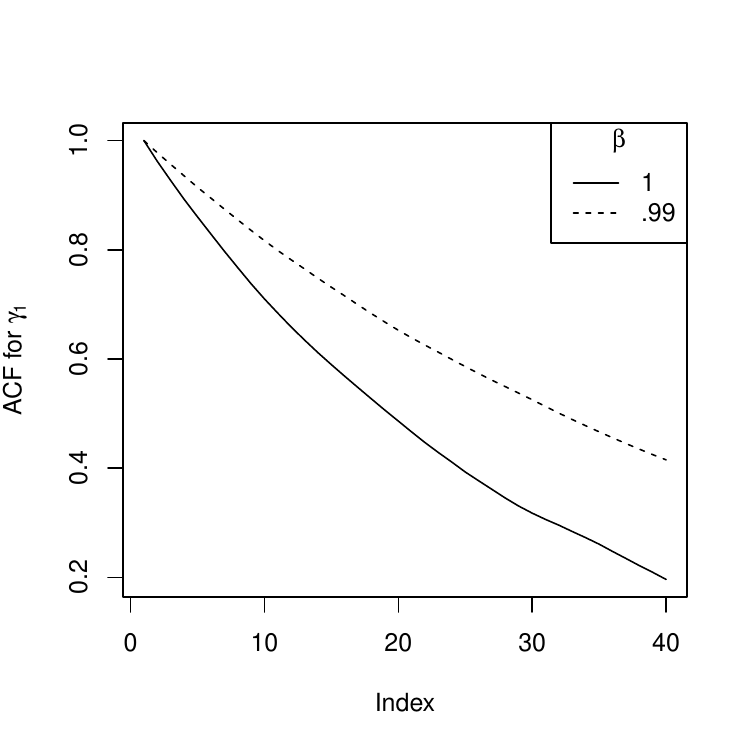}
		\includegraphics[width = 2.1in]{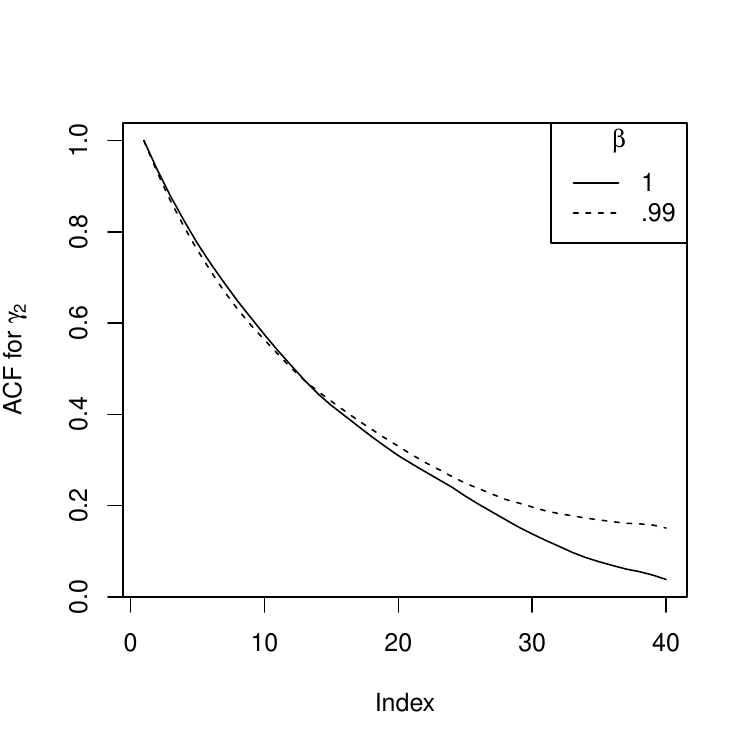}\vspace{-.5cm}
	\caption{Autocorrelation plots for $\gamma_1$ and $\gamma_2$.}
	\label{fig:wf_acf}
\end{figure}
 The estimated ESS for the posterior mean of $\gamma$ is 255 and 197, respectively for Barker's and portkey Barker's. However, the ESS per hour is 7.36 and 56.45, respectively. Thus, the portkey Barker's algorithm is $\approx$ 7 times more efficient that the original Barker's algorithm.
%
We stress that implementation of the portkey two-coin algorithm requires only minor changes to the code for any two-coin algorithm.

\section{Discussion} 
\label{sec:discussion}

Motivated by Bernoulli factories, we introduce a family of MCMC acceptance probabilities for intractable target distributions\footnote{An \texttt{R} package that implements the Bernoulli factories is here: \texttt{https://github.com/dvats/portkey} and complete reproducible codes for our examples can be found here: \texttt{https://github.com/dvats/PortkeyPaperCode}.}. We argue that the (flipped) portkey two-coin algorithm is a robust alternative to the two-coin algorithm; when the two-coin algorithm is efficient, the portkey two-coin algorithm is essentially similar to the two-coin algorithm. However, when the two-coin algorithm is computationally infeasible, the portkey two-coin algorithm is significantly more stable, has a smaller mean execution time, and allows for larger proposal variances. The resulting Markov chain is exact and for values of $\beta$ close to 1, the loss in statistical efficiency is small compared to the gain in computational efficiency. Moreover, tuning the portkey Barker's sampler is far easier since feedback from the algorithm is significantly faster than the original Barker's sampler.

Finally, the (portkey) two-coin algorithm finds immediate use in Bayesian inference for diffusions and jump-diffusions, as illustrated in Section~\ref{sec:wright_fisher_diffusion} and Bayesian models with intractable priors. To implement these algorithms more generally requires reasonable bounds on the target distribution. 
We believe the solution to this is to construct local bounds compatible with these algorithms and this makes for important future work.

\section{Acknowledgements}

Dootika Vats was supported by the National Science Foundation (DMS/1703691). Fl\'{a}vio Gon\c{c}alves
thanks FAPEMIG, CNPq and the University of Warwick for financial support. Krzysztof {\L}atuszy\'nski is supported by the Royal Society through the  Royal Society University Research Fellowship.  Gareth Roberts is supported by the EPSRC grants: {\em ilike}
(EP/K014463/1), CoSInES (EP/R034710/1) and Bayes for Health (EP/R018561/1).

\singlespacing
\bibliographystyle{apalike}
\bibliography{mcref}
\end{document}